\documentclass[a4paper, 11pt,reqno]{amsart}

\usepackage{amssymb}
\usepackage{amsmath}
\usepackage{amsthm}
\usepackage{latexsym}
\usepackage[T1]{fontenc}
\usepackage{mathrsfs}
\usepackage{bm}
\usepackage{tipa}
\usepackage{slashed}
\usepackage{mathtools}
\usepackage{enumitem}
\usepackage{tikz-cd}
\usepackage{color}
\usepackage{stackrel}
\usepackage[boxsize=0.5em,aligntableaux=top]{ytableau}

\usepackage{soul}

\usepackage[numbers]{natbib}
     
\numberwithin{equation}{section}

\usepackage{hyperref}

\hypersetup{
	colorlinks =true,
	linkcolor =blue,
	pdftitle={},
	citecolor=blue,
	filecolor=blue,
	urlcolor=blue,
}

\newcounter{mnotecount}[section]

\allowdisplaybreaks[2]

\theoremstyle{plain}
\newtheorem{theorem}{Theorem}
\newtheorem{proposition}[theorem]{Proposition}
\newtheorem{lemma}[theorem]{Lemma}

\newtheorem{remark}[theorem]{Remark}

\renewcommand{\d}{{\rm d}}
\renewcommand{\i}{{\rm i}}

\title{Generalized Siklos space-times}

\author[B. Araneda]{Bernardo Araneda}
\address{Max Planck Institut f\"ur Gravitationsphysik (Albert Einstein Institute), Am M\"uhlenberg 1, D-14476 Potsdam, Germany}
\email{bernardo.araneda@aei.mpg.de}
\author[\'A. J. Murcia]{\'Angel J. Murcia}
\address{INFN, Sezione di Padova, Via Francesco Marzolo 8, 35131 Padova, Repubblica Italiana}
\email{angel.murcia@pd.infn.it}

\begin{document}

\date{\today} 
\begin{abstract}
Motivated by supersymmetry methods in general relativity, we study four-dimensional Lorentzian space-times with a complex Dirac spinor field satisfying a Killing-spinor-like equation where the Killing constant is promoted to a complex function. We call the resulting geometry a generalized Siklos space-time. 
After deriving a number of identities for complex spaces, we specialize to Lorentz signature, where we show that the Killing function must be real and that the corresponding Dirac spinor is Majorana (as long as the space-time is not conformally flat), and we obtain the local form of the metric. 
We show that the purely gravitational degrees of freedom correspond to waves, whereas the matter sources generically correspond, via Einstein's field equations, to a sum of pure radiation and a space-like perfect fluid. Consequently, we conclude that the physically relevant case is obtained when the Killing function is homogeneous on the wave surfaces.
\end{abstract}

\maketitle

\section{Introduction}

In general relativity (GR), Einstein's field equations --- including possibly a non-vanishing stress-energy tensor --- constitute an intricate system of second-order non-linear partial differential equations on a Lorentzian manifold. Given their intrinsic complexity, the discovery of novel solutions with physically sensible matter content is of interest to both the GR and theoretical high-energy physics communities. As a consequence, the study of geometric methods which lead to the reduction or simplification of such equations is thus of paramount importance. One of these techniques has been widely employed in the supergravity literature and is that of \emph{supersymmetry}, where, broadly speaking, one imposes a set of first-order differential equations which generically imply the fulfilment of \emph{almost} all  equations of motion for gravity and the corresponding matter fields \cite{Ivanov}.

More precisely, a classical supergravity solution --- in the sense that all fermionic fields have been truncated, which is always consistent, cf. \cite[Ch. 22]{Freedman} ---  is called supersymmetric or BPS if it preserves some supersymmetry, in the sense that it is invariant under the fermionic transformation rules determined by the corresponding supergravity theory. At the classical level, these transformation rules take the form of a first order differential operator applied to a spinor field (cf. \cite[Eq. (22.56)]{Freedman}), and the solution is supersymmetric if the result of this operation is zero. The spinor field satisfying the corresponding differential equation is called a \emph{super-covariantly constant} spinor. The simplest cases correspond to a parallel spinor (e.g. \cite{Tod1983,Gibbons1986,Tod1995,Byrant2000,Murcia20,A22}) and a Killing spinor (e.g. \cite{Siklos1985,Gibbons1986,Friedrich}).

More general spinors have also been considered in the literature. Here we will be interested in Lorentzian four-manifolds with spinor fields satisfying the equation
\begin{align}
 \nabla_{X}\psi = \tfrac{\lambda}{\sqrt{2}} X\cdot\psi, \qquad \lambda\in C^{\infty}(M),
 \label{KSE0}
\end{align}
for all vector fields $X$, where $\cdot$ denotes Clifford product and $\lambda$ is a complex function. The case in which $\lambda$ is constant corresponds to the ordinary Killing spinor equation. We may refer to solutions of \eqref{KSE0} as {\em Killing spinors} with {\em Killing function} $\lambda$.\footnote{In \cite{Rademacher}, such spinor fields were called \emph{generalized Killing}, while in other works, see e.g. \cite{BarGauduchonMoroianu,FK1,FK2}, the notion of {\em generalized Killing spinors} is more general.} 
The case $\lambda={\rm const.}$ was studied by Siklos \cite{Siklos1985}, who called the solutions {\em Lobatchevski waves}. We will then say that a space-time with a solution to \eqref{KSE0} is a \emph{generalized Siklos space-time}. 
Equation \eqref{KSE0} --- for non-constant $\lambda$ --- has been studied in the Riemannian setup, see for example \cite{Rademacher}. Here we will focus on Lorentzian (and complex) space-times, which is a topic of modern research --- cf. \cite{Shahbazi24}, where real spinors parallel with respect to a metric connection with totally skew-symmetric torsion are studied in Lorentzian three-manifolds, which is equivalent to \eqref{KSE0} in three space-time dimensions. Apart from its mathematical interest, one may argue that \eqref{KSE0} represents a natural generalization of the usual supersymmetry condition in four-dimensional minimal AdS supergravity, and is part of the program of using first-order spinorial differential equations to find particular solutions of second-order differential equations (the Einstein equations with appropriate matter content). 

Our main results may be summarized as follows. After providing the necessary background material in Section \ref{sec:prelimi}, we obtain in Section \ref{sec:CKS} a number of identities for a complex space admitting a solution to \eqref{KSE0}, see Lemma \ref{lemma:CKS} and Remark \ref{remark:complex}. We also briefly show that in Euclidean signature, any 4-space with a solution to \eqref{KSE0} must be conformally flat (Prop. \ref{prop:riemannian}). 
In Section \ref{sec:KSL} we impose Lorentzian reality conditions. We show that if the space-time is not conformally flat, then a complex solution to \eqref{KSE0} reduces to a real one. We obtain the local form of the metric and curvature, and show that the purely gravitational degrees of freedom correspond to waves. We find an adapted coordinate system, and prove that the Killing function can only depend on the coordinate labeling the wave surfaces --- see \cite[Sec. 18.1]{Grif-Pod} for their definition in generic Kundt space-times --- and on one of the coordinates that are transverse to the propagation of the waves. These results are contained in Theorem \ref{result:lorentzian} and Prop. \ref{prop:CF}. 
In Section \ref{sec:Einstein} we give a physical interpretation, by analyzing the stress-energy tensor that sources the gravitational field via Einstein's field equations. We find that in the generic case, the matter content corresponds to null dust plus a space-like perfect fluid. The physically relevant case occurs when the Killing function is homogeneous on the wave surfaces.

We will use 2-spinor techniques and the {\em abstract-index} notation, cf. \cite{PR1, PR2}. See also \cite{HuggettTod} and \cite[Ch. 13]{Wald} for background on 2-spinors.

\section{Preliminaries}

\label{sec:prelimi}

In this section we give an elementary review of those aspects of the 2-spinor formalism that are needed for this work, including its connection with Dirac spinors. Our notation and conventions  follow \cite{PR1, PR2} (cf. in particular the appendix in \cite{PR2}). 

\subsection{Spinor algebra} 

Let $V$ be a $d$-dimensional {\em complex} vector space, and $g_{mn}$ a symmetric non-degenerate bilinear form in $V$. We identify $V\cong V^{*}$ via $g_{mn}$. Indices $a,b,...,m,n,...$ are {\em abstract}. The Clifford algebra is 
\begin{align}
\gamma_{m}\gamma_{n}+\gamma_{n}\gamma_{m}=-2g_{mn}\mathbb{I}. \label{Clifford}
\end{align}
The Dirac matrices $\gamma_{m}$ are abstract linear transformations $\mathbb{D}\to V^{*}\otimes\mathbb{D}$, where $\mathbb{D}$ is a complex vector space whose elements are called Dirac spinors, and $\mathbb{I}$ is the identity in $\mathbb{D}$. The dimension of $\mathbb{D}$ is $2^{\lfloor d \rfloor/2}$. If $d$ is even, $\mathbb{D}$ splits into two irreducible pieces as $\mathbb{D}=\mathbb{S}\oplus\mathbb{S}'$ (see \cite[eq. (B.16)]{PR2}), where $\mathbb{S}$ and $\mathbb{S}'$ are $2^{(\frac{d}{2}-1)}$-dimensional and are said to have opposite chirality. 

We are interested in $d=4$, so $V$ and $\mathbb{D}$ are 4-dimensional, and $\mathbb{S}$, $\mathbb{S'}$ are 2-dimensional. Elements of $\mathbb{S}$, $\mathbb{S}'$ are called Weyl spinors. In abstract indices, we denote generic elements of $\mathbb{S}$ and $\mathbb{S}'$ by $o_{A}$ and $\tilde{o}_{A'}$, respectively. A generic Dirac spinor $\psi\in\mathbb{D}$ is decomposed as
\begin{align}
\psi = \left( \begin{matrix} o_{A} \\ \tilde{o}_{A'} \end{matrix} \right). \label{DS}
\end{align}
The spaces $\mathbb{S}$ and $\mathbb{S'}$ correspond to the fundamental representations $(\frac{1}{2},0)$ and $(0,\frac{1}{2})$ of ${\rm SL}(2,\mathbb{C})$, and the Dirac representation is $(\frac{1}{2},0)\oplus(0,\frac{1}{2})$.

The spaces $\mathbb{S}$, $\mathbb{S}'$ are equipped with symplectic forms $\epsilon_{AB}=-\epsilon_{BA}$, $\epsilon_{A'B'}=-\epsilon_{B'A'}$, respectively. The inverses are denoted $\epsilon^{AB}$, $\epsilon^{A'B'}$. These objects are used to raise and lower indices as follows:
\begin{align}
 o^{A}=\epsilon^{AB}o_{B}, \qquad \iota_{A}=\epsilon_{BA}\iota^{B}, \qquad
 \tilde{o}^{A'} = \epsilon^{A'B'}\tilde{o}_{B'}, \qquad \tilde\iota_{A'}=\epsilon_{B'A'}\tilde\iota^{B'}. 
 \label{RSI}
\end{align}
Note that the ``square'' of any 2-spinor is zero: $o_Ao^A=\epsilon^{AB}o_Ao_B=0$, etc. 

An unprimed spin dyad is a pair of spinors $o_{A},\iota_{A}$ such that $\epsilon^{AB}o_{A}\iota_{B}=o_{A}\iota^{A}=1$, and a primed spin dyad is $\tilde{o}_{A'},\tilde\iota_{A'}$ with $\tilde{o}_{A'}\tilde\iota^{A'}=1$. It follows that $\epsilon_{AB}=o_{A}\iota_{B}-\iota_{A}o_{B}$ and $\epsilon_{A'B'}=\tilde{o}_{A'}\tilde\iota_{B'}-\tilde\iota_{A'}\tilde{o}_{B'}$. 

The Clebsch-Gordan decomposition gives $(\frac{1}{2},0)\otimes(0,\frac{1}{2})\cong(\frac{1}{2},\frac{1}{2})$. The latter is the vector representation, so $\mathbb{S}\otimes\mathbb{S}'$ and $V$ are isomorphic. The isomorphism is denoted $\sigma$. In abstract indices:
\begin{align}
 u_{AA'} \mapsto u_{a} = \sigma_{a}{}^{AA'}u_{AA'}, 
 \qquad
 v_{a} \mapsto v_{AA'} \equiv \sigma^{a}{}_{AA'}v_{a}.
 \label{isom}
\end{align}
The map $\sigma$ must be compatible with the action of the orthogonal group. One way to realize this is as follows: think of elements of $V\cong\mathbb{C}^4$ as column vectors, elements of $\mathbb{S}\otimes\mathbb{S}'\cong\mathbb{C}^2\otimes\mathbb{C}^2$ as $2\times2$ matrices, then take the canonical basis $e_0,...,e_3$ of $\mathbb{C}^4$, and set $\sigma_{0}=\sigma(e_{0})=\frac{1}{\sqrt{2}}\mathbb{I}_{2\times2}$ and $\sigma_{i}=\sigma(e_{i})=\frac{1}{\sqrt{2}}\sigma'_{i}$, where $\sigma'_{i}$ are the Pauli matrices ($i=1,2,3$). 

It follows that $g(v,v)=2\det\sigma(v)$ for any $v\in V$, which implies that the metric is equivalent to the object $g_{AA'BB'} \equiv \epsilon_{AB}\epsilon_{A'B'}$.
We also see that $v$ is null iff $\det\sigma(v)=0$. This means that $v$ is null iff it corresponds to a product of spinors, $\sigma(v)=o\otimes\tilde{o}$, or in abstract indices $v_{AA'}=o_A\tilde{o}_{A'}$. In particular, the spin dyads $\{o_{A},\iota_{A}\}$, $\{\tilde{o}_{A'},\tilde\iota_{A'}\}$ give origin to a null tetrad:
\begin{align}
\ell_{a}:=\sigma_{a}{}^{AA'}o_{A}\tilde{o}_{A'}, \quad n_{a}:=\sigma_{a}{}^{AA'}\iota_{A}\tilde{\iota}_{A'}, \quad m_{a}:=\sigma_{a}{}^{AA'}o_{A}\tilde{\iota}_{A'}, \quad 
\tilde{m}_{a}:=\sigma_{a}{}^{AA'}\iota_{A}\tilde{o}_{A'}. 
\label{nulltetrad}
\end{align}
All of these vectors are null, and they satisfy $\ell_a n^a=1$, $m_a\tilde{m}^a=-1$, with the rest of the contractions being zero.
In practice, it is customary (and will be done in this manuscript) to omit the symbol $\sigma$ in expressions like \eqref{isom} and \eqref{nulltetrad} (cf. \cite{PR1}), so that we have the abstract-index equivalences 
\begin{align}\label{indices}
 a \equiv AA', \qquad b\equiv BB', \qquad c\equiv CC', \qquad ... 
\end{align}

Following Penrose \& Rindler (cf. \cite[footnote on pp. 221]{PR1} and \cite[eq. (B.86)]{PR2}),
we will use the following representation of Dirac matrices and their action on Dirac spinors:
\begin{align}
 \gamma_{m} = \sqrt{2} 
 \left( \begin{matrix} 0 & \epsilon_{MA}\delta_{M'}^{B'} \\
 \epsilon_{M'A'}\delta_{M}^{B} & 0 \end{matrix} \right), 
 \qquad
 \gamma_{m}\psi = \sqrt{2}  
\left(   \begin{matrix}  \epsilon_{MA}\tilde{o}_{M'} \\
 \epsilon_{M'A'}o_{M}  \end{matrix}  \right).
 \label{Diracmatrices}
\end{align}

\begin{remark}\label{Remark:AltRep}
An alternative representation of Dirac matrices satisfying \eqref{Clifford} is 
\begin{align}
 \gamma'_{m} = \sqrt{2} 
 \left( \begin{matrix} 0 & \i\epsilon_{MA}\delta_{M'}^{B'} \\
 -\i\epsilon_{M'A'}\delta_{M}^{B} & 0 \end{matrix} \right), 
 \qquad 
 \gamma'_{m}\psi = \sqrt{2}  
 \left(   \begin{matrix}  \i\epsilon_{MA}\tilde{o}_{M'} \\
 -\i\epsilon_{M'A'}o_{M}  \end{matrix}  \right).
 \label{Diracmatrices2}
\end{align}
In this paper we will use \eqref{Diracmatrices}, but we note that the difference between \eqref{Diracmatrices} and \eqref{Diracmatrices2} is important when analyzing the reality properties of $\lambda$ in \eqref{KSE0}, see e.g. Remark \ref{remark:lambdareal}.
\end{remark}

\subsection{Spinor fields}
\label{sec:spinorfields}

Let $(M,g_{ab})$ be a complex 4-dimensional spin manifold, with spin structure $P_{\rm spin}$. The unprimed spinor bundle is defined as the associated vector bundle $\mathcal{S}=P_{\rm spin}\times_{(\frac{1}{2},0)}\mathbb{C}^{2}$, and similarly the primed spinor bundle is $\mathcal{S}'=P_{\rm spin}\times_{(0,\frac{1}{2})}\mathbb{C}^{2}$. The bundle of Dirac spinors is $\mathcal{D}=\mathcal{S}\oplus\mathcal{S}'$. These bundles generalize, respectively, the spin spaces $\mathbb{S}$, $\mathbb{S}'$ and $\mathbb{D}$. A spinor field is a smooth section of the corresponding spinor bundle. The relation $\mathbb{S}\otimes\mathbb{S}'\cong V$ becomes now $\mathcal{S}\otimes\mathcal{S}'\cong TM$. The Levi-Civita connection $\nabla_{a}$ lifts to a connection on the various spinor bundles, which we also denote by $\nabla_{a}$. Using the convention \eqref{indices}, we have $\nabla_{a}\equiv\nabla_{AA'}$. 

Our convention for the Riemann tensor is $[\nabla_{a},\nabla_{b}]v^{d}=R_{abc}{}^{d}v^{c}$ (cf. \cite[Eq. (4.2.30)]{PR1}). Lowering an index $R_{abcd}=g_{de}R_{abc}{}^{e}$, we have $R_{abcd}=R_{[ab][cd]}$. The additional symmetries $R_{abcd}=R_{cdab}$ and $R_{[abc]d}=0$ give the following spinor decomposition:
\begin{subequations}\label{decompRiemann}
\begin{align}
\nonumber R_{abcd} \equiv R_{AA'BB'CC'DD'} ={}& X_{ABCD}\epsilon_{A'B'}\epsilon_{C'D'}+\Phi_{ABC'D'}\epsilon_{A'B'}\epsilon_{CD} \\
& + \Phi_{A'B'CD}\epsilon_{AB}\epsilon_{C'D'}+\tilde{X}_{A'B'C'D'}\epsilon_{AB}\epsilon_{CD}, \label{Riemann} \\
X_{ABCD} ={}& \Psi_{ABCD} + \Lambda(\epsilon_{AC}\epsilon_{BD}+\epsilon_{AD}\epsilon_{BC}), \label{X} \\
\tilde{X}_{A'B'C'D'} ={}& \tilde{\Psi}_{A'B'C'D'} + \Lambda(\epsilon_{A'C'}\epsilon_{B'D'}+\epsilon_{A'D'}\epsilon_{B'C'}), \label{tildeX}
\end{align}
\end{subequations}
where $\Psi_{ABCD}=\Psi_{(ABCD)}$ and $\tilde{\Psi}_{A'B'C'D'}=\tilde{\Psi}_{(A'B'C'D')}$ correspond respectively to the anti-self-dual (ASD) and self-dual (SD) Weyl tensors (cf. \cite[eq. (4.6.42)-(4.6.43)]{PR1}), and we say that they represent the {\em purely gravitational} degrees of freedom. The quantities $\Phi_{ABC'D'}=\Phi_{(AB)(C'D')}$ and $\Lambda$ correspond to the trace-free Ricci tensor and Ricci scalar, respectively \cite[eq. (4.6.22)-(4.6.23)]{PR1}. We have:
\begin{equation}\label{curvaturespinors}
\begin{aligned}
& C^{-}_{abcd} = \Psi_{ABCD}\epsilon_{A'B'}\epsilon_{C'D'}, \qquad
C^{+}_{abcd} = \tilde\Psi_{A'B'C'D'}\epsilon_{AB}\epsilon_{CD}, \\
& \Phi_{ABA'B'} = -\tfrac{1}{2}(R_{ab} - \tfrac{R}{4}g_{ab}), \qquad
\Lambda = \tfrac{1}{24}R.
\end{aligned}
\end{equation}

The {\em Petrov classification} refers to the algebraic type of the Weyl curvature: at a generic point $x\in M$, there exist four spinors $\alpha_{A},\beta_{A},\gamma_{A},\delta_{A}$ (which can coincide), called {\em principal spinors}, such that (see \cite[eq. (8.1.1)]{PR2})
\begin{align}
\Psi_{ABCD} = \alpha_{(A}\beta_{B}\gamma_{C}\delta_{D)}.
\end{align}
Depending on how many of the principal spinors coincide, we can have different Petrov types: see \cite[eq. (8.1.6)]{PR2}. An analogous classification applies to $\tilde\Psi_{A'B'C'D'}$. The two classifications are independent for complex spaces and for Euclidean signature, but they coincide for Lorentz signature.

\subsection{Conformal transformations}

In later sections we will work with conformally re-scaled geometries, so it is convenient to collect  some formulas here.
Let $\Omega$ be a nowhere vanishing scalar field. A conformal transformation is a rescaling $g_{ab}\mapsto \hat{g}_{ab}:=\Omega^{2}g_{ab}$. The Levi-Civita connection of $\hat{g}_{ab}$ will be denoted by $\hat{\nabla}_{a}$. The relation between $\nabla_{a}$ and $\hat\nabla_{a}$ when acting on generic spinor fields $\mu_{A},\nu_{A'}$ is (cf. \cite[eq. (5.6.15)]{PR1}):
\begin{align}
 \hat\nabla_{AA'}\mu_{B} = \nabla_{AA'}\mu_{B} - \Upsilon_{A'B}\mu_{A}, 
 \qquad
 \hat\nabla_{AA'}\nu_{B'} = \nabla_{AA'}\nu_{B'} - \Upsilon_{AB'}\nu_{A'},
 \label{CTspinors}
\end{align}
where $\Upsilon_{a}:=\partial_{a}\log\Omega$. Both of the Weyl curvature spinors are conformally invariant, and the formulas for the transformations of the Ricci spinor $\Phi_{ABA'B'}$ and the scalar curvature $\Lambda$ are given in \cite[eq. (6.8.24)-(6.8.25)]{PR2}. In summary:
\begin{subequations}\label{CT}
\begin{align}
\hat\Psi_{ABCD} ={}& \Psi_{ABCD}, 
\qquad \hat{\tilde{\Psi}}_{A'B'C'D'} = \tilde{\Psi}_{A'B'C'D'}, \label{CTWeyl} \\
\hat\Phi_{ab} ={}& \Phi_{ab} - \hat\nabla_{a}\Upsilon_{b}-\Upsilon_{a}\Upsilon_{b} 
+ \tfrac{1}{4}\hat{g}_{ab}\hat{g}^{cd}(\hat\nabla_{c}\Upsilon_{d}+\Upsilon_{c}\Upsilon_{d}), 
\label{CTRspinor} \\
\Omega^{2}\hat\Lambda ={}& \Lambda - \tfrac{1}{4}\Omega^{3}\hat{g}^{ab}\hat\nabla_{a}\hat\nabla_{b}\Omega^{-1}.
\label{CTRscalar}
\end{align}
\end{subequations}

\section{Complex Killing spinors}
\label{sec:CKS}

Let $(M,g_{ab})$ be a 4-dimensional complex space. In abstract indices, equation \eqref{KSE0} is
\begin{align}
\nabla_{m}\psi = \tfrac{\lambda}{\sqrt{2}} \gamma_{m}\psi,
\qquad \lambda \in C^\infty(M). \label{KSE}
\end{align}
Choosing the representation \eqref{Diracmatrices} of Dirac matrices, and renaming indices for convenience, a short calculation shows that \eqref{KSE} is equivalent to 
\begin{subequations}\label{KSE2}
\begin{align}
 \nabla_{AA'}o_{B} ={}& \lambda \, \epsilon_{AB}\tilde{o}_{A'}, \label{KSE21} \\
 \nabla_{AA'}\tilde{o}_{B'} ={}& \lambda \, \epsilon_{A'B'}o_{A}. \label{KSE22}
\end{align}
\end{subequations}
We summarize the main results of this section in the following:
\begin{lemma}\label{lemma:CKS}
Let $(M,g_{ab})$ be a complex space admitting a non-trivial solution to the system \eqref{KSE2}. Then:
\begin{enumerate}
\item\label{item:Weyl} 
Both Weyl curvature spinors $\Psi_{ABCD}$ and $\tilde\Psi_{A'B'C'D'}$ are Petrov type N, with principal spinors $o_{A}$ and $\tilde{o}_{A'}$ respectively.
\item\label{item:Ricci} 
The trace-free Ricci spinor $\Phi_{ABA'B'}$ and the scalar curvature $\Lambda$ satisfy
\begin{subequations}
\begin{align}
& \Phi_{ABA'B'}o^{B} = -\tilde{o}_{(B'}\nabla_{A')A}\lambda, \label{riccio} \\
& \Phi_{ABA'B'}\tilde{o}^{B'} = - o_{(B}\nabla_{A)A'}\lambda, \label{ricciotilde} \\
& \tilde{o}^{A}\nabla_{AA'}\lambda + 2(\Lambda + \lambda^2)o_{A} = 0, \label{SCo} \\
& o^{A}\nabla_{AA'}\lambda + 2(\Lambda + \lambda^2)\tilde{o}_{A'} = 0. \label{SCotilde} 
\end{align}
\end{subequations}
\item\label{item:congruence} 
The null vector $\ell^{a}=o^{A}\tilde{o}^{A'}$ is: geodesic, divergence-free, twist-free, shear-free, and Killing.
\item\label{item:conformal} 
$(M,g_{ab})$ is conformal to a complex space with a parallel unprimed 2-spinor, and to another (generally different) complex space with a parallel primed 2-spinor.
\end{enumerate}
\end{lemma}

\begin{proof} 
We start with item \eqref{item:Weyl}.
Applying another derivative in one of the equations in \eqref{KSE2}, and using the other equation, we get
\begin{subequations}
\begin{align}
 \nabla_{AA'}\nabla_{BB'}o_{C} ={}& (\nabla_{AA'}\lambda)\, \epsilon_{BC} \tilde{o}_{B'}+ \lambda^{2}\epsilon_{BC}\epsilon_{A'B'}o_{A}, \label{dd1} \\
 \nabla_{AA'}\nabla_{BB'}\tilde{o}_{C'} ={}& (\nabla_{AA'}\lambda)\, \epsilon_{B'C'} o_{B}+ \lambda^{2}\epsilon_{B'C'}\epsilon_{AB}\tilde{o}_{A'}. \label{dd2}
\end{align}
\end{subequations}
Contracting \eqref{dd1} with $\epsilon^{A'B'}$ and symmetrizing over $ABC$, and doing something analogous with \eqref{dd2}, after using the curvature decomposition \eqref{decompRiemann} we get
\begin{subequations}
\begin{align}
\Psi_{ABCD}o^{D} &=\tilde\Psi_{A'B'C'D'}\tilde{o}^{D'}=0,
\end{align}
\end{subequations}
which implies that 
\begin{equation}\label{RestrictionsWeyl}
 \Psi_{ABCD} = F \, o_{A}o_{B}o_{C}o_{D}, \qquad 
\tilde\Psi_{A'B'C'D'} = \tilde{F} \, \tilde{o}_{A'}\tilde{o}_{B'}\tilde{o}_{C'}\tilde{o}_{D'},
\end{equation}
for some complex functions $F, \tilde{F}$. Hence both $\Psi_{ABCD} $ and $\tilde\Psi_{A'B'C'D'}$ are type N. 

\smallskip
For item \eqref{item:Ricci}, contraction of \eqref{dd1} with $\epsilon^{AB}$ and symmetrization over $A'B'$ leads to \eqref{riccio}, while contraction of \eqref{dd1} with $\epsilon^{A'B'}$ and symmetrization over $AB$ leads to \eqref{SCo}. Analogous operations applied to \eqref{dd2} lead to \eqref{ricciotilde} and \eqref{SCotilde}. 

\smallskip
Next we focus on item \eqref{item:congruence}. From \eqref{KSE2} we deduce 
$\nabla_{AA'}o^{A} = 2\lambda \tilde{o}_{A'}$ and $\nabla_{AA'}\tilde{o}^{A'} = 2\lambda o_{A}$. It then follows that $\nabla_{a}\ell^{a}=\nabla_{AA'}(o^{A}\tilde{o}^{A'})=0$, so $\ell^a$ is divergence-free. Contracting \eqref{KSE21}  with $o^{A}o^{B}$ and \eqref{KSE22} with $\tilde{o}^{A'}\tilde{o}^{B'}$, we immediately get 
$o^{A}o^{B}\nabla_{AA'}o_{B} = 0 = \tilde{o}^{A'}\tilde{o}^{B'}\nabla_{AA'}\tilde{o}_{B'}$,
which are the necessary and sufficient conditions for the null congruence to be geodesic and shear-free \cite[Section 7.3]{PR2}. From \eqref{KSE2} we also get
\begin{align}
\nabla_{a}\ell_{b} = \lambda(\tilde{o}_{A'}\tilde{o}_{B'}\epsilon_{AB}+o_{A}o_{B}\epsilon_{A'B'}),
\label{nablaell}
\end{align} 
so $\nabla_{a}\ell_{b}=\nabla_{[a}\ell_{b]}$, which gives $\nabla_{(a}\ell_{b)}=0$ and so, $\ell_{a}$ is Killing. Finally, from the above expression for $\nabla_{a}\ell_{b}$ we also deduce $(\nabla_{[a}\ell_{b]})(\nabla^{a}\ell^{b})=0$, so the congruence is twist-free (cf. \cite[eq. (7.1.27) and pp. 178]{PR2}).

\smallskip
Finally, let us prove item \eqref{item:conformal}. Consider a conformal transformation $\hat{g}_{ab}=\Omega^2g_{ab}$, and let $\hat{o}_{A}:=\Omega{o}_{A}$. Using the first equation in \eqref{CTspinors}, and also eq. \eqref{KSE21}, we have
\begin{align}
\hat\nabla_{AA'}\hat{o}_{B}=\Omega \, \epsilon_{AB}(\lambda\tilde{o}_{A'}+o^{C}\nabla_{A'C}\log\Omega).
\label{hatds}
\end{align}
Thus, $\hat\nabla_{AA'}\hat{o}_{B}=0$ if and only if $o^{A}\nabla_{AA'}f=\alpha_{A'}$, where $f\equiv\log\Omega$ and $\alpha_{A'}\equiv-\lambda\tilde{o}_{A'}$. The necessary and sufficient conditions for this equation to admit a solution $f$ are given in \cite[Lemma (7.3.20)]{PR2} (recall that $o^{A}$ is geodesic shear-free): they are $o^{A}o^{B}\nabla_{A}^{A'}\alpha_{A'}=\alpha_{A'}o^{A}\nabla_{A}^{A'}o^{B}$. Using \eqref{KSE2}, it is straightforward to show that this equation is satisfied (both sides are independently zero), so a solution $f$ exists, and thus the space $(M,\hat{g}_{ab})$ has a parallel unprimed spinor $\hat{o}_{A}$. An analogous calculation, using a generally different conformal factor $\tilde\Omega$, shows that the space $(M,\tilde{g}_{ab})$ with $\tilde{g}_{ab}=\tilde\Omega^2g_{ab}$ has a parallel primed spinor $\tilde\Omega\tilde{o}_{A'}$.
\end{proof}

\begin{remark}\label{remark:complex}
We make note of the following facts.
\begin{itemize}
\item If $\lambda$ is constant, the trace-free Ricci spinor is null, $\Phi_{ab} = G\ell_a\ell_b$ for some function $G$, and the scalar curvature satisfies $\Lambda=-\lambda^2$.
\item From \eqref{riccio} (or \eqref{ricciotilde}), we see that the Einstein condition $\Phi_{ABA'B'}=0$ can be satisfied only if $\lambda$ is constant, in which case the problem reduces to the standard Killing spinor equation.
\item\label{item:kundt} 
From item \ref{item:congruence} in Lemma \ref{lemma:CKS} we see that $(M,g_{ab})$ is a complex Kundt space (i.e. a complex space with a null geodesic congruence that is expansion-free, twist-free, and shear-free).
\item\label{item:PS}  
In \cite[Prop. 2.2]{A22}, the structure of the metric and curvature of a complex space with a parallel 2-spinor is given. The geometry depends on a single complex function $\Theta$, in terms of which the line element can be expressed as in \cite[eq. (2.13)]{A22}. The space from Lemma \ref{lemma:CKS} is conformal to this geometry.
\end{itemize}
\end{remark}

\subsection{Riemannian signature}
We can specialize Lemma \ref{lemma:CKS} to Riemannian signature by introducing an appropriate spinor conjugation \cite{Woodhouse}. This is a quaternionic operation $\dagger$ that preserves chirality. A spinor $\varphi_{A}$ and its complex conjugate $\varphi^{\dagger}_{A}$ are linearly independent, and analogously for primed spinors. The conjugation extends to higher-valence spinors, where it is involutive in the case of an even number of indices. The two Petrov classifications are independent, but $\Psi_{ABCD}$ and $\tilde\Psi_{A'B'C'D'}$ are invariant under $\dagger$, which means that there is no Petrov type N. As a direct application of Lemma \ref{lemma:CKS}, we have:

\begin{proposition}\label{prop:riemannian}
If a Riemannian 4-space $(M,g_{ab})$ admits a non-trivial solution $(o_A,\tilde{o}_{A'})$ to \eqref{KSE2} with Killing function $\lambda$, then it is conformally flat, and the Ricci spinor is 
\begin{align}\label{ricciriemannian}
 \Phi_{ABA'B'} = \frac{1}{(o_Co^{\dagger C})} \left[ o^{\dagger}_{B}\tilde{o}_{(B'}\nabla_{A')A}\lambda - o_{B}\tilde{o}^{\dagger}_{(B'}\nabla_{A')A}\bar\lambda \right].
\end{align}
If $\lambda$ is purely imaginary, then it is constant and $(o_A,\tilde{o}_{A'})$ is an ordinary Killing spinor.
\end{proposition}

\begin{proof}
The conformally flat condition follows from item \ref{item:Weyl} in Lemma \ref{lemma:CKS}: since there is no Petrov type N, both Weyl spinors must vanish. 
Formula \eqref{ricciriemannian} follows from eq. \eqref{riccio} and its complex conjugate, which is $\Phi_{ABA'B'}o^{\dagger B}=-\tilde{o}^{\dagger}_{(B'}\nabla_{A')A}\bar\lambda$. Contracting \eqref{ricciriemannian} with $\epsilon^{AB}$ and using that $\Phi_{ABA'B'}=\Phi_{(AB)A'B'}$, we get:
\begin{align}
 0 = \tilde{o}_{(B'}o^{\dagger A}\nabla_{A')A}\lambda - \tilde{o}^{\dagger}_{(B'}o^{A}\nabla_{A')A}\bar\lambda.
\end{align}
If $\lambda$ is purely imaginary, $\bar\lambda=-\lambda$, then using \eqref{SCotilde} the above equation reduces to $0=-4(\Lambda+\lambda^2)\tilde{o}_{(A'}\tilde{o}^{\dagger}_{B')}$, which implies $\Lambda+\lambda^2=0$ and, using \eqref{SCotilde} again, gives $\lambda={\rm const}$.
\end{proof}

\begin{remark}\label{remark:lambdareal}
In \cite{Rademacher}, it is mentioned that the case in which the problem reduces to the ordinary Killing spinor equation corresponds to $\lambda$ {\em real}. The difference with Prop. \ref{prop:riemannian} originates in our choice of representation of Dirac matrices \eqref{Diracmatrices}, see Remark \ref{Remark:AltRep}.
\end{remark}

\section{Lorentzian signature}
\label{sec:KSL}

Following \cite{PR1,PR2}, our convention for Lorentz signature is $(+---)$.
Here the spinor conjugation \cite{Woodhouse}, denoted by an overbar, is involutive, but it interchanges chirality, so that if $o_{A}\in\mathbb{S}$ then $\overline{(o_{A})}\equiv\bar{o}_{A'}\in\mathbb{S}'$, etc. This operation allows to define a {\em Majorana} (or {\em real}) spinor as a Dirac spinor \eqref{DS} such that $\tilde{o}_{A'}=\bar{o}_{A'}$.
Since the Riemann tensor is real, from \eqref{decompRiemann} we deduce that 
\begin{align}
\tilde{\Psi}_{A'B'C'D'}=\overline{(\Psi_{ABCD})}, \qquad
\overline{(\Phi_{A'B'CD})}=\Phi_{ABC'D'}, \qquad \bar{\Lambda} = \Lambda.
\label{RCL}
\end{align}
Thus, $\Phi_{ab}$ and $\Lambda$ are real, and $\Psi_{ABCD}$ and $\tilde{\Psi}_{A'B'C'D'}\equiv\bar\Psi_{A'B'C'D'}$ are complex-conjugates, so there is only one Petrov classification. The results of Lemma \ref{lemma:CKS} now specialize to:

\begin{theorem}\label{result:lorentzian}
Let $(M,g_{ab})$ be a Lorentzian space-time admitting a non-trivial solution $(o_{A},\tilde{o}_{A'})$ to \eqref{KSE2}, with Killing function $\lambda$. Then:
\begin{enumerate}
\item\label{item:lorentzian}
If $(M,g_{ab})$ is not conformally flat, then $\lambda$ is real and $(o_{A},\tilde{o}_{A'})$ can be taken to be a Majorana spinor.
\item\label{item:ppwave}
$(M,g_{ab})$ is conformal to a real Lorentzian space-time with a parallel 2-spinor.
\item\label{item:metric} 
There exist two real coordinates $u,v$, a complex coordinate $\zeta$, and two real functions $H$ and $\Omega$ such that the line element of the space-time metric $g_{ab}$ is 
\begin{align}
 \d{s}^2 = 2\Omega^{-2}(\d{u}\d{v}-\d\zeta\d\bar\zeta) + H\d{v}^2. \label{STmetric}
\end{align}
\end{enumerate}
\end{theorem}

\begin{proof}
For item \eqref{item:lorentzian}, suppose that $(M,g_{ab})$ is not conformally flat. Then using the reality conditions \eqref{RCL} in \eqref{RestrictionsWeyl}, it follows that 
\begin{align}
 \tilde{o}_{A'} = \mu \bar{o}_{A'}, \qquad \tilde{F}\mu^4= \bar{F},  
 \label{RCLKS}
\end{align}
for some scalar field $\mu$. Now, the complex conjugate of \eqref{KSE21} is 
\begin{align}
 \nabla_{AA'}\bar{o}_{B'}=\bar\lambda\epsilon_{A'B'}\bar{\tilde{o}}_{A}. \label{CCKSE21}
\end{align}
From \eqref{KSE22} and \eqref{CCKSE21} we see that $\nabla_{A(A'}\tilde{o}_{B')}=0$ and $\nabla_{A(A'}\bar{o}_{B')}=0$. Using $\tilde{o}_{A'} = \mu \bar{o}_{A'}$ we then get $\bar{o}_{(B'}\nabla_{A')A}\mu=0$, which implies $\nabla_{a}\mu=0$ and thus $\mu$ is constant. Therefore: 
\begin{align}
\lambda\epsilon_{A'B'}o_{A} = \nabla_{AA'}\tilde{o}_{B'}=\mu\nabla_{AA'}\bar{o}_{B'} 
= \mu\bar\lambda\epsilon_{A'B'}\bar{\tilde{o}}_{A} = \bar\lambda|\mu|^2\epsilon_{A'B'}o_{A}
\end{align}
where in the first equality we used \eqref{KSE22}, in the second \eqref{RCLKS}, in the third \eqref{CCKSE21}, and in the fourth $\bar{\tilde{o}}_{A}=\bar{\mu}o_{A}$. It follows that 
$\lambda=\bar\lambda|\mu|^2$.
Writing $\lambda=Re^{\i\theta}$, we get $e^{2\i\theta}=|\mu|^2\in\mathbb{R}_{>0}$, which implies $\theta=0\mod\pi$, thus $\lambda\in\mathbb{R}$. It also follows that $|\mu|^2=1$, so $\mu=e^{\i\alpha}$ for some real constant $\alpha$. 
Defining $\phi_{A}:=e^{-\i\alpha/2}o_{A}$, we get $\nabla_{AA'}\phi_{B}=\lambda\epsilon_{AB}\bar{\phi}_{A'}$. Since $\lambda$ is real, the Dirac spinor $\psi'=(\phi_{A},\bar{\phi}_{A'})^{\rm t}$ is a solution to \eqref{KSE2} and is Majorana.

For item \ref{item:ppwave}, note that although the result is similar to item \eqref{item:conformal} in Lemma \ref{lemma:CKS}, it does not automatically follow from there since we must show that the scalar field $f=\log\Omega$ used in the proof of that lemma can now be chosen to be real. From \eqref{hatds}, $f$ must satisfy
\begin{align}\label{EqForRealCF}
 o^{A}\nabla_{AA'}f + \lambda\bar{o}_{A'} = 0.
\end{align}
If a real solution $f$ to this equation exists, then the proof of Lemma \ref{lemma:CKS} applies and the spinor $\hat{o}_{A}=\Omega o_{A}$ is parallel, with $\Omega$ real. Thus we can focus on showing that a real $f$ solving \eqref{EqForRealCF} can be found. To do this, we will use the fact that the null congruence defined by $\ell^a=o^{A}\bar{o}^{A'}$ is twist-free. From \cite[pp. 179-180]{PR2}, a null congruence $k^{a}$ is twist-free iff it is hyper-surface orthogonal iff $k_a$ is proportional to a gradient. Since we showed that $\ell^{a}$ is twist-free, it follows that there must exist a {\em real} scalar field $\mathcal{F}$ such that $\mathcal{F}\ell_{a}$ is a closed 1-form, i.e. $\nabla_{[a}(\mathcal{F}\ell_{b]})=0$. To see the equation that $\mathcal{F}$ satisfies, we first note the general identity 
\begin{align*}
\nabla_{[a}(\mathcal{F}\ell_{b]}) = \tfrac{1}{2}\nabla_{C(A'}[\mathcal{F}\ell^{C}_{B')}]\epsilon_{AB} + \tfrac{1}{2}\nabla_{C'(A}[\mathcal{F}\ell^{C'}_{B)}]\epsilon_{A'B'}.
\end{align*}
Since $\mathcal{F}$ and $\ell^a$ are real, the two terms in the right side are complex conjugates, so $\nabla_{[a}(\mathcal{F}\ell_{b]})=0$ if and only if $\nabla_{C(A'}[\mathcal{F}\ell^{C}_{B')}]=0$. 
From \eqref{nablaell} we see that $\nabla_{A(A'}\ell^{A}_{B')}=2\lambda \bar{o}_{A'}\bar{o}_{B'}$. Using $\ell^{a}=o^{A}\bar{o}^{A'}$, we then have
\begin{align}
\nonumber \nabla_{A(A'}[\mathcal{F}\ell^{A}_{B')}] ={}& 
\ell^{A}_{(B'}\nabla_{A')A}\mathcal{F} + \mathcal{F}\nabla_{A(A'}\ell^{A}_{B')} \\
\nonumber ={}& \bar{o}_{(B'}o^{A}\nabla_{A')A}\mathcal{F} + 2\lambda \, \mathcal{F} \, \bar{o}_{A'}\bar{o}_{B'} \\
={}& 2\mathcal{F}\bar{o}_{(B'}\left[o^{A}\nabla_{A')A}\log\sqrt{\mathcal{F}} + \lambda \bar{o}_{A')} \right].
\end{align}
The last line vanishes if and only if the term inside the square brackets vanishes. Thus we can take $f=\log\sqrt{F}$ in \eqref{EqForRealCF} (which means $\Omega=\sqrt{\mathcal{F}}$).

Finally, concerning item \eqref{item:metric}, notice that $\hat{o}_{A}=\Omega o_{A}$ is a parallel spinor for the metric $\hat{g}_{ab}=\Omega^2 g_{ab}$, and since $\Omega$ is real, $\hat{g}_{ab}$ is a Lorentzian pp-wave. In appendix \ref{appendix:ppwaves} we recall from \cite{A22} some identities for pp-waves, including the definition of the coordinates $u,v,\zeta,\bar\zeta$ (see \eqref{coordpp}) and the line element \eqref{ppwave}. Then \eqref{STmetric} follows from \eqref{ppwave} and $g_{ab}=\Omega^{-2}\hat{g}_{ab}$, setting $H=\Omega^{-2}\hat{H}$.
\end{proof}

\begin{remark}
From the above we deduce the following:
\begin{itemize}
\item If $(M,g_{ab})$ is conformally flat, then \eqref{RCLKS} does not necessarily hold, and $(o_{A},\tilde{o}_{A'})$ is not necessarily Majorana. 
\item If $(M,g_{ab})$ is not conformally flat, \eqref{KSE2} is equivalent to the single equation
\begin{align}
 \nabla_{AA'}o_{B}  = \lambda \, \epsilon_{AB}\bar{o}_{A'}, \qquad \lambda=\bar{\lambda} \in C^\infty(M).
 \label{realKSE}
\end{align}
\item The type N property of the Weyl tensor and its conformal invariance allow us to interpret the gravitational degrees of freedom as waves, cf. \cite[Sec. 9.7]{PR2}. The coordinates in \eqref{STmetric} are well-adapted to this: $u$ is an affine parameter along the rays, the hypersurfaces of constant $v$ are the wave surfaces, and $\zeta$ is a complex coordinate transverse to the propagation of the waves. 
\end{itemize}
\end{remark}

To find an expression for the conformal factor $\Omega$ in \eqref{STmetric}, we recall that $\Omega$ is defined by eq. \eqref{EqForRealCF} with $f=\log\Omega$:
\begin{align}\label{EqOmega}
o^A\partial_{AA'}\Omega=-\lambda\Omega\bar{o}_{A'} 
\end{align}
To solve this equation, we use identities \eqref{relationspindyads}, \eqref{nulltetradpp} and \eqref{NVpp}. Defining
\begin{align}\label{coordxy}
 x:=\zeta+\bar\zeta, \qquad y:=\i(\bar\zeta-\zeta),
\end{align}
and noticing that $\partial_{\zeta} = \partial_x-\i\partial_y$, $\partial_{\bar\zeta} = \partial_x+\i\partial_y$, it follows from \eqref{EqOmega} that
\begin{align}\label{dOmega}
 \partial_{u}\Omega=0, \qquad \partial_{x}\Omega=\lambda, \qquad \partial_{y}\Omega=0,
\end{align}
so $\Omega=\Omega(x,v)$. Integrating \eqref{dOmega}, we get:
\begin{proposition}\label{prop:CF}
With the definition \eqref{coordxy}, the Killing function satisfies $\lambda=\lambda(x,v)$, and the conformal factor $\Omega$ in \eqref{STmetric} is 
\begin{align}\label{Omega}
 \Omega(x,v) = \int\lambda(x,v)\d{x} + f(v),
\end{align}
where $f(v)$ is an arbitrary real function of $v$.
\end{proposition}

\begin{remark}\label{remark:freedomCF}
From \eqref{coordpp}, we see that $\Omega$ is defined only up to multiplication by functions of $v,\zeta$: since $\hat{o}^{A}\partial_{AA'}v=0=\hat{o}^{A}\partial_{AA'}\zeta$, if $F=F(v,\zeta)$ then $F\Omega$ also satisfies \eqref{EqOmega}. This means that $(F\Omega) o_{A}$ is a parallel spinor for the metric $(F\Omega)^2g_{ab}$. Choosing $F$ to be real implies $F=F(v)$; then the metric $(F\Omega)^2g_{ab}$ will also be a real, Lorentzian pp-wave.
\end{remark}

We note that some simplifications are obtained in the special case 
\begin{align}\label{lambdav}
 \lambda=\lambda(v). 
\end{align}
In the expression \eqref{Omega} for $\Omega$, define $\tilde{x}$ by $\d{x}=\d\tilde{x}-\d{f}/\lambda$, then we can eliminate $f(v)$. Dropping the tilde, we have $\Omega(x,v)=\lambda(v)x$. But from Remark \ref{remark:freedomCF} we know that we have the freedom to multiply $\Omega$ by any function of $v$. Multiplying then by $b\lambda^{-1}$, where $b\neq0$ is an arbitrary constant, and denoting the new conformal factor again by $\Omega$, we have $\Omega = bx$, and the space-time metric \eqref{STmetric} is 
\begin{align}\label{siklos}
 \d{s}^2 = \frac{1}{b^2x^2}(2\d{u}\d{v}-\d{x}^2-\d{y}^2+\hat{H}\d{v}^2).
\end{align}
In \cite[Eq. (3)]{Siklos1985}, Siklos obtained this form of the metric for space-times admitting an {\em ordinary} Killing spinor, i.e. a solution to \eqref{realKSE} with $\lambda={\rm const}$. Here we see that \eqref{siklos} is valid for space-times with Killing spinors \eqref{KSE2} with Killing function $\lambda=\lambda(v)$.

\section{Einstein's equations}
\label{sec:Einstein}

Let $(M,g_{ab})$ be the Lorentzian space-time from Theorem \ref{result:lorentzian} (i.e., a generalized Siklos space-time). In this section we give a physical interpretation of $(M,g_{ab})$, by computing the Einstein's field equations and analyzing the corresponding energy-momentum tensor: 
\begin{align}\label{EFE}
 G_{ab} + cg_{ab} = - 8\pi G T_{ab}.
\end{align}
Our conventions follow \cite[Eq. (4.6.30)]{PR1}, so $c$ is the cosmological constant, $T_{ab}$ is the energy-momentum tensor of the sources. We took the Newton's constant to be $\frac{1}{8\pi}$.

\subsection{The Einstein tensor}

We use the null tetrad $(\ell_a,n_a,m_a,\bar{m}_a)$ given by \eqref{nulltetradpp}-\eqref{nulltetradK}. It is convenient to also define the space-like vector 
\begin{align}
 X_{a}:=m_a+\bar{m}_a.
\end{align}
In terms of these variables, we have:
\begin{proposition}
The Einstein tensor of $(M,g_{ab})$ is
\begin{equation}
\begin{aligned}
G_{ab} ={}& -\left[\Omega(\hat{H}_{xx}+\hat{H}_{yy})+2(\Omega_{vv}-\Omega_{x}\hat{H}_{x}) \right]\Omega^3\ell_a\ell_b \\
& -2\Omega\left[ \Omega_{xx}\, X_aX_b+2\,\Omega\,\Omega_{vx} \, \ell_{(a}X_{b)} \right]
+(6\Omega_{x}^2-4\Omega\Omega_{xx})g_{ab}. \label{Einsteintensor}
\end{aligned}
\end{equation}
\end{proposition}

\begin{proof}
Using \cite[Eq. (4.6.25)]{PR1}, the Einstein tensor is in general
\begin{align}
G_{ab}=-2\Phi_{ab}-6\Lambda g_{ab}. \label{Einsteinaux}
\end{align}
Expressions for $\Phi_{ab}$ and $\Lambda$ can be obtained from identities \eqref{CTRspinor}, \eqref{CTRscalar}, and \eqref{ppcurvature}. We have
\begin{align}
\Phi_{ab} ={}& \tfrac{\Omega^{4}}{2}\hat{H}_{\zeta\bar\zeta}\ell_a\ell_b + \hat\nabla_{a}\Upsilon_{b} + \Upsilon_{a}\Upsilon_{b} - \tfrac{1}{4}\hat{g}_{ab}\hat{g}^{cd}(\hat\nabla_{c}\Upsilon_{d}+\Upsilon_{c}\Upsilon_{d}), \label{Ricciaux} \\
\Lambda ={}& \tfrac{1}{4}\Omega^{3}\hat{g}^{ab}\hat\nabla_{a}\hat\nabla_{b}\Omega^{-1}, \label{SCaux}
\end{align}
where we recall that $\Upsilon_{a}=\partial_{a}\log\Omega$. Using \eqref{NVpp}, \eqref{CCpp} and \eqref{Omega}, we find
\begin{align*}
\hat\nabla_{a}\Upsilon_{b} + \Upsilon_{a}\Upsilon_{b} = 
\Omega^{3}(\Omega_{vv}-\Omega_{x}\hat{H}_{x})\ell_a\ell_b + \Omega\,\Omega_{xx}\, X_a X_b + 2\Omega^2\Omega_{vx}\,\ell_{(a}X_{b)}.
\end{align*}
Contracting with $\hat{g}^{ab}=\Omega^{-2}g^{ab}$ and using $g^{ab}X_aX_b=-2$, we get
\begin{align*}
\hat{g}^{ab}(\hat\nabla_{a}\Upsilon_{b} + \Upsilon_{a}\Upsilon_{b}) = -2\Omega^{-1}\Omega_{xx}.
\end{align*}
Replacing in \eqref{Ricciaux}:
\begin{align*}
\Phi_{ab} = \left[\tfrac{\Omega}{2}\hat{H}_{\zeta\bar\zeta} + \Omega_{vv}-\Omega_{x}\hat{H}_{x}\right]\Omega^3 \ell_a\ell_b + \Omega\,\Omega_{xx}\, X_a X_b + 2\Omega^2\Omega_{vx}\,\ell_{(a}X_{b)} + \frac{\Omega}{2}\Omega_{xx} g_{ab}.
\end{align*}
For the scalar curvature $\Lambda$, we use \eqref{SCaux} and \eqref{waveoperatorpp}:
\begin{align}
\Lambda = -\tfrac{1}{2}\Omega^{3}\partial_{x}^2\Omega^{-1} 
= \tfrac{1}{2}\Omega\,\Omega_{xx} - \Omega_{x}^2.
\end{align}
Replacing the above expressions for $\Phi_{ab}$ and $\Lambda$ in \eqref{Einsteinaux}, we get \eqref{Einsteintensor}.
\end{proof}

\subsection{The special case $\lambda=\lambda(v)$}
\label{sec:speciallambda}
Recall from the discussion below eq. \eqref{lambdav} that in the case $\lambda=\lambda(v)$, the conformal factor can be chosen to be $\Omega=bx$ where $b$ is an arbitrary constant. Replacing in \eqref{Einsteintensor}, it immediately follows that
\begin{align}\label{Einsteinspeciallambda}
 G_{ab} = -\left[\hat{H}_{xx} + \hat{H}_{yy} - \frac{2}{x}\hat{H}_{x} \right]b^4x^4\ell_a\ell_b + 6b^2g_{ab}.
\end{align}
Comparing to \eqref{EFE}, we see that there is a cosmological term (with cosmological constant $c=-6b^2$) and an energy-momentum tensor corresponding to pure radiation (null dust). This is analogous to the interpretation of an ordinary Siklos space-time \cite{Siklos1985}, as expected since we showed that the metric is \eqref{siklos}, even though $\lambda$ is not constant but $\lambda=\lambda(v)$.

\subsection{The generic case $\lambda=\lambda(x,v)$}
Since $\lambda=\Omega_{x}$ (recall \eqref{dOmega}), the assumption $\lambda_{x}\neq0$ implies $\Omega_{xx} \neq 0$. We can then rearrange eq. \eqref{Einsteintensor} as $G_{ab}=-(T^{(\ell)}_{ab}+T^{(s)}_{ab})$, where
\begin{align}
 T^{(\ell)}_{ab} = \rho_{(\ell)}\ell_{a}\ell_{b}, \qquad
 T^{(s)}_{ab} = (\rho_{(s)}+p_{(s)})s_{a}s_{b} - p_{(s)}g_{ab},
\end{align}
and
\begin{align}
\rho_{(\ell)} ={}& 2\Omega^{3}\left[\tfrac{\Omega}{2}(\hat{H}_{xx}+\hat{H}_{yy}) + \Omega_{vv} - \Omega_{x}\hat{H}_{x}-\frac{\Omega_{vx}^2}{\Omega_{xx}} \right], \\
p_{(s)} ={}& 6\Omega_{x}^2-4\Omega\Omega_{xx}, \\
\label{eq:eosgen}
\rho_{(s)} ={}& -p_{(s)} + 2\frac{\Omega}{\Omega_{xx}}, \\
s_{a} ={}& \Omega\Omega_{vx}\, \ell_{a} + \Omega_{xx}\, X_{a}.
\end{align}
Comparing to \eqref{EFE}, we see that the energy-momentum tensor $T_{ab}$ is the superposition of a pure radiation term $T^{(\ell)}_{ab}$ and a perfect fluid term $T^{(s)}_{ab}$. The equation of state \eqref{eq:eosgen} associated with this last piece generalizes that of a cosmological constant ($p=-\rho$) with the addition of the term $2\frac{\Omega}{\Omega_{xx}}$. However, the vector field $s_{a}$ satisfies $s_as^a=-2\Omega_{xx}^2<0$, so it is {\em space-like}. Since $s_{a}$ represents the 4-velocity of the fluid, this means that it propagates faster than light. We then conclude that the case $\lambda_{x}\neq0$ corresponds to an unphysical situation.

As an example, take $\lambda=x$, which via \eqref{Omega} gives $\Omega=\frac{1}{2}x^2$ (we set $f(v)=0$). Then the variables of the space-like fluid are $p_{(s)}=4x^2$, $\rho_{(s)}=-3x^2$, $s_{a}=X_{a}$.

\subsection{The Einstein case}
From Remark \ref{remark:complex}, the Einstein condition $\Phi_{ab}=0$ (i.e. $R_{ab}=c g_{ab}$) can only be satisfied if $\lambda$ is constant.\footnote{An analogous result was obtained in three Lorentzian dimensions, cf. \cite[Corollary 4.9]{Shahbazi24}.}. This is a special case of the situation in Section \ref{sec:speciallambda}: the Einstein tensor for $\lambda={\rm const.}$ is given by \eqref{Einsteinspeciallambda}. So the Einstein condition simply reduces to the vanishing of the term inside the square brackets in \eqref{Einsteinspeciallambda}:
\begin{equation}
\hat{H}_{xx}+\hat{H}_{yy} - \frac{2}{x}\hat{H}_{x}=0.
\label{EinsteinEqSK}
\end{equation}
This result was also obtained by Siklos \cite{Siklos1985}, who found the general solution to \eqref{EinsteinEqSK}:
\begin{equation}
\hat{H}=x^2 \partial_{x} \left ( x^{-1}(f(\zeta,v)+\bar{f}(\bar{\zeta},v))\right) ,
\label{SolSiklos}
\end{equation}
where $f(\zeta,v)$ is an arbitrary function (cf. \cite[Eq. (6)]{Siklos1985}).

\section{Concluding remarks}

We proved that a four-dimensional Lorentzian space-time admitting a solution to the supersymmetry-motivated equation \eqref{KSE0} must be given by eq. \eqref{STmetric}, and we showed that it can be interpreted as describing gravitational waves where the matter sources are generically given by null dust plus a space-like perfect fluid, generalizing the standard Siklos space-times \cite{Podolsky1998}. We showed that the Killing function $\lambda$ depends only on two real coordinates, one that labels the wave surfaces and one that is transverse to them. We concluded that the physically interesting case corresponds to the situation in which $\lambda$ is homogenous on the surfaces, since then the space-like fluid is absent and the sources are pure radiation. We also saw that the Einstein case reduces to the ordinary Killing spinor equation, where the Einstein equation linearizes \eqref{EinsteinEqSK} and its general solution can be found \eqref{SolSiklos}.

Regarding future directions, it would be interesting to study which generalizations of the Killing-spinor condition lead to Lorentzian four-manifolds with physically sensible matter content. Another direction to examine is the study of Killing spinors with a generic Killing function in higher space-time dimensions. Also, it would be intriguing to explore the holographic interpretation of (generalized) Siklos space-times.

\subsection*{Acknowledgements}
The authors would like to thank C. S. Shahbazi for his very useful comments. BA acknowledges support of the Institut Henri Poincaré (UAR 839 CNRS-Sorbonne Université), and LabEx CARMIN (ANR-10-LABX-59-01). Work of \'AJM has been supported by a postdoctoral grant from the Istituto Nazionale di Fisica Nucleare, Bando 23590. \'AJM acknowledges the additional financial support from the International Network on Quantum Fields and Strings to visit the Max Planck Institut (MPI) f\"ur Gravitationsphysik (Albert Einstein Institute), where this project was initiated under the very warm hospitality of the MPI. \'AJM would like to dedicate this manuscript in loving memory of Carmencita Pérez Quesada.

\appendix

\section{Pp-waves}
\label{appendix:ppwaves}

A four-dimensional Lorentzian space-time $(M,\hat{g}_{ab})$ with a parallel 2-spinor $\hat{o}_{A}$ is a pp-wave. The structure of its metric and curvature can be seen e.g. from \cite[Prop. 2.1]{A22}: there are two real coordinates $u,v$, and a complex coordinate $\zeta$, defined by
\begin{align}\label{coordpp}
 \partial_{a}v=\hat{o}_{A}\bar{\hat{o}}_{A'}, \qquad 
 \partial_{a}\zeta=\hat{o}_{A}\bar{\hat{\iota}}_{A'}, \qquad 
 (\partial_{u})^{a} = \hat{o}^{A}\bar{\hat{o}}^{A'}, \qquad 
 (\partial_{\bar\zeta})^{a} = - \hat{o}^{A}\bar{\hat{\iota}}^{A'},
\end{align}
where $\hat\iota_{A}$ is a spinor field with $\hat{o}_{A}\hat{\iota}^{A}=1$, such that the line element of $\hat{g}_{ab}$ is
\begin{align}\label{ppwave}
\widehat{\d{s}}{}^2 = 2(\d{u}\d{v}-\d\zeta\d\bar\zeta) + \hat{H}\d{v}^2,
\end{align}
for some real function $\hat{H}$. The spin dyad $\{\hat{o}_{A},\hat\iota_{A}\}$ and its complex conjugate give a natural null tetrad via \eqref{nulltetrad}:
\begin{align}
 \hat\ell^a=\hat{o}^{A}\bar{\hat{o}}^{A'}, \qquad 
 \hat{n}^a=\hat{\iota}^{A}\bar{\hat{\iota}}^{A'}, \qquad 
 \hat{m}^a=\hat{o}^{A}\bar{\hat{\iota}}^{A'}, \qquad
 \bar{\hat{m}}^a=\hat{\iota}^{A}\bar{\hat{o}}^{A'},
 \label{nulltetradpp}
\end{align}
satisfying $\hat\ell_a\hat{n}^a=1=-\hat{m}_a\hat{\bar{m}}^a$. Expressions for these vectors in terms of the coordinates $(u,v,\zeta,\bar\zeta)$, and for the connection coefficients, 
can be found in \cite[eq. (A.1)]{A22}:
\begin{align}
& \hat\ell^{a}\partial_{a} = \partial_{u}, \qquad \hat{m}^{a}\partial_{a} = -\partial_{\bar\zeta}, \qquad \bar{\hat{m}}^{a}\partial_{a} = -\partial_{\zeta}, \qquad 
\hat{n}^{a}\partial_{a} = \partial_{v} - \tfrac{1}{2}\hat{H}\partial_{u}, \label{NVpp} \\
& \hat\nabla_{a}\hat\ell^b = 0, \quad 
\hat\nabla_{a}\hat{m}^{b} = -\bar{\hat{\kappa}}'\hat\ell_{a}\hat\ell^b, \quad
\hat\nabla_{a}\bar{\hat{m}}^{b} = -\hat{\kappa}'\hat\ell_{a}\hat\ell^b, \quad 
\hat\nabla_{a}\hat{n}^{b} = -\hat\ell_{a}(\hat{\kappa}'\hat{m}^{b}+\bar{\hat{\kappa}}'\bar{\hat{m}}^{b}), \label{CCpp}
\end{align}
where $\hat\kappa'=\tfrac{1}{2} \hat{H}_{\zeta}$. 
The curvature is 
\begin{align}
\hat\Phi_{ab} = \tfrac{1}{2}\hat{H}_{\zeta\bar\zeta}\hat\ell_a\hat\ell_b, 
 \qquad \hat\Lambda = 0, 
 \qquad \hat\Psi_{ABCD} = \tfrac{1}{2}\hat{H}_{\bar\zeta\bar\zeta}\hat{o}_{A}\hat{o}_{B}\hat{o}_{C}\hat{o}_{D}, \label{ppcurvature}
\end{align}
The wave operator acting on an arbitrary scalar field $\varphi$ is \cite[Eq. (2.10)]{A22}:
\begin{align}\label{waveoperatorpp}
 \hat{g}^{ab}\hat\nabla_{a}\hat\nabla_{b}\varphi=2(\varphi_{uv}-\varphi_{\zeta\bar\zeta}) - \hat{H}\varphi_{uu}.
\end{align}

\medskip
The space-time $(M,g_{ab})$ from Sections \ref{sec:KSL} and \ref{sec:Einstein} is conformally related to the pp-wave, via $g_{ab}=\Omega^{-2}\hat{g}_{ab}$. The relation between the spin dyads of $g_{ab}$ and $\hat{g}_{ab}$ is
\begin{align}\label{relationspindyads}
 \hat{o}_{A} = \Omega o_{A}, \qquad \hat\iota_{A}=\iota_A, \qquad \hat{o}^A=o^A, \qquad \hat\iota^A=\Omega^{-1}\iota^A,
\end{align}
and analogously for $\{\bar{\hat{o}}_{A'},\bar{\hat\iota}_{A'}\}$. The null tetrads are related by 
\begin{equation}\label{nulltetradK}
\begin{aligned}
  \ell^a ={}& \hat\ell^a, \qquad & n^a={}&\Omega^{2}\hat{n}^a, \qquad & m^a={}& \Omega\hat{m}^a, \qquad & \bar{m}^a ={}& \Omega\bar{\hat{m}}^{a}, \\
 \ell_a ={}& \Omega^{-2}\hat\ell_a, \qquad & n_a={}&\hat{n}_a, \qquad & m_a ={}& \Omega^{-1}\hat{m}_a, \qquad & \bar{m}_a ={}& \Omega^{-1}\bar{\hat{m}}_{a},
\end{aligned}
\end{equation}

\section{The local isometry type of Siklos space-times}

It was shown in the main text that any (generalized) Siklos space-time is locally isometric to:
\begin{align}
 \d{s}^2 = 2\Omega^{-2}(\d{u}\d{v}-\d\zeta\d\bar\zeta) + H\d{v}^2, 
 \qquad H:=\Omega^{-2}\hat{H}, \label{metricKSLapp}
\end{align}
so it depends on two real functions: $H$ and $\Omega$. In particular, standard Siklos space-times are such that $\Omega=\lambda x$, where $\lambda \in \mathbb{R}$ is the Killing constant. On the other hand, in the literature, it was proven  (cf. \cite[Corollary 2.21]{Murcia22}) that every Siklos space-time (these space-times were called supersymmetric Kundt in \cite{Murcia22}) was locally isometric to:
\begin{equation}
\d\tilde{s}^2=\tilde{H} \d \tilde{v}^2+e^{\tilde{\mathcal{F}}}\left ( \d \tilde{u}+\tilde{\beta}-\frac{e^{-\tilde{\mathcal{F}}}}{4\tilde{\lambda}^2} \partial_{\tilde{v}}  \tilde{\mathcal{F}} \d \tilde{\mathcal{F}}-\partial_{\tilde{v}} \tilde{\mathcal{G}}\d \tilde{\mathcal{G}}\right) \odot \d \tilde{v}+\frac{e^{\tilde{\mathcal{F}}}}{4\tilde{\lambda}^2} \left ( e^{-\tilde{\mathcal{F}}} \d \tilde{\mathcal{F}}^2+4 \tilde{\lambda}^2 \d \tilde{\mathcal{G}}^2\right),
\label{eq:metms}
\end{equation}
where $\tilde{H}, \tilde{\mathcal{F}}$  and $\tilde{\mathcal{G}}$ are functions independent of the $\tilde{u}$-coordinate, $\tilde{\lambda}$ is a constant and $\tilde{\beta}$ is a one-form transverse to the $(\tilde{u},\tilde{v})$-coordinates whose components do not depend on the $\tilde{v}$-coordinate. The one-form  $\tilde{\beta}$ must satisfy a certain condition (cf. \cite[eq. (2.20)]{Murcia22}), which we opt to present later in the exposition for clarity. Also, observe that we are using a tilde to affect those objects that refer to the local characterization of Siklos space-times given in \cite{Murcia22}. 

At this point, it is necessary to make compatible the results of the present manuscript with those of \cite{Murcia22}. To this aim, note that \eqref{metricKSLapp} and \eqref{eq:metms} will be locally isometric if the following condition holds:
\begin{equation}
\d \tilde{\alpha}= \d \tilde{\Gamma} \wedge \d \tilde{v}\,, \quad \tilde{\alpha}=\tilde{\beta}-\frac{e^{-\tilde{\mathcal{F}}}}{4\tilde{\lambda}^2} \partial_{\tilde{v}} \tilde{\mathcal{F}} \d \tilde{\mathcal{F}}-\partial_{\tilde{v}} \tilde{\mathcal{G}}\d \tilde{\mathcal{G}},
\label{eq:comparar}
\end{equation}
where $\tilde{\Gamma}$ is a certain function independent of $\tilde{u}$. Indeed, in such a case, define $\tilde{H}=H-2e^{\tilde{\mathcal{F}}} \tilde{\Gamma}$. Then, after introducing a new coordinate $\d u=\d \tilde{u}+\tilde{\alpha}-\tilde{\Gamma}\d \tilde{v}$, we arrive to \eqref{metricKSLapp} by setting $e^{\tilde{\mathcal{F}}}=\Omega^{-2}$. Consequently, all our efforts should be focused on proving \eqref{eq:comparar}. For that, let us split:
\begin{equation}
\tilde{\alpha}=\tilde{\kappa}-\frac{e^{-\tilde{\mathcal{F}}}}{4 \tilde{\lambda}^2} \left(\partial_{\tilde{v}} \tilde{\mathcal{F}} \right)^2 \d \tilde{v}-\left (\partial_{\tilde{v}} \tilde{\mathcal{G}} \right)^2 \d \tilde{v} \,, \quad \tilde{\kappa}=\tilde{\beta}-\frac{e^{-\tilde{\mathcal{F}}}}{4\tilde{\lambda}^2} \partial_{\tilde{v}} \tilde{\mathcal{F}} \d_X \tilde{\mathcal{F}}-\partial_{\tilde{v}} \tilde{\mathcal{G}}\d_X \tilde{\mathcal{G}},
\end{equation}
where $\d_X$ denotes the exterior derivative along the space transverse to the coordinates $(\tilde{u},\tilde{v})$. Since the difference between $\tilde{\alpha}$ and $\tilde{\kappa}$ lies along the $\tilde{v}$-coordinate, to guarantee that \eqref{eq:comparar} holds it suffices to observe that $\d \tilde{\kappa}= \d_X \tilde{\Gamma} \wedge \d \tilde{v}$. We write:
\begin{equation}
\d \tilde{\kappa}=-\partial_{\tilde{v}}  \tilde{\kappa} \wedge \d \tilde{v}+\d_X \tilde{\kappa}\,, \quad  \d_X \tilde{\kappa}=\d_X \tilde{\beta}-\frac{e^{-\tilde{\mathcal{F}}}}{4 \tilde{\lambda}^2} \partial_{\tilde{v}} \d_X \tilde{\mathcal{F}} \wedge \d_X \mathcal{F} -\partial_{\tilde{v}} \d_X \tilde{\mathcal{G}} \wedge \d_X \tilde{\mathcal{G}},
\end{equation}
where $\partial_{\tilde{v}}  \tilde{\kappa}$ stands for the one-form transverse to the $(\tilde{u},\tilde{v})$-space whose components are the derivatives with respect to $\tilde{v}$ of those of $\tilde{\kappa}$. Let us now focus momentarily on $\d_X \tilde{\kappa}$. On the one hand, an appropriate and careful massaging of eq. (2.20) of \cite{Murcia22} reveals that:
\begin{equation}
\d_X \tilde{\beta}=\frac{1}{4 \tilde{\lambda}^2} \left ( -\left\langle \partial_{\tilde{v}} \d_X \tilde{\mathcal{F}}, \d_X \tilde{\mathcal{G}} \right\rangle_{\tilde{q}} +\left\langle \partial_{\tilde{v}} \d_X \tilde{\mathcal{G}}, \d_X \tilde{\mathcal{F}} \right\rangle_{\tilde{q}} \right) \d_X \tilde{\mathcal{F}} \wedge \d_X \tilde{\mathcal{G}},
\end{equation}
where $\langle \cdot, \cdot \rangle_{\tilde{q}}$ denotes the pairing of one-forms along the space transverse to the coordinates $(\tilde{u},\tilde{v})$ with the metric $\tilde{q}=\frac{1}{4 \tilde{\lambda^2}} \d_X \tilde{\mathcal{F}} \otimes \d_X \tilde{\mathcal{F}}+ e^{\tilde{\mathcal{F}}}\d_X \tilde{\mathcal{G}}\otimes \d_X \tilde{\mathcal{G}}$. On the other hand:
\begin{align}
-\frac{e^{-\tilde{\mathcal{F}}}}{4 \tilde{\lambda}^2} \partial_{\tilde{v}} \d_X \tilde{\mathcal{F}} \wedge \d_X \tilde{\mathcal{F}}&=\frac{1}{4 \tilde{\lambda}^2} \left\langle \partial_{\tilde{v}} \d_X \tilde{\mathcal{F}}, \d_X \tilde{\mathcal{G}} \right\rangle_{\tilde{q}} \d_X \tilde{\mathcal{F}} \wedge \d_X \tilde{\mathcal{G}},\\
-\partial_{\tilde{v}} \d \tilde{\mathcal{G}} \wedge \d_X \tilde{\mathcal{G}}&=-\frac{1}{4\tilde{\lambda}^2} \left\langle \partial_{\tilde{v}} \d_X \tilde{\mathcal{G}}, \d_X \tilde{\mathcal{F}} \right\rangle_{\tilde{q}}\d_X \tilde{\mathcal{F}} \wedge \d_X \tilde{\mathcal{G}}.
\end{align}
Taking into account that the operations $\d_X$ and $\partial_{\tilde{v}}$ commute, by summing up terms we conclude that $\d_X \tilde{\kappa}=0$. Now, by the same token, we have that $\partial_{\tilde{v}} \d_X \tilde{\kappa}= \d_X \partial_{\tilde{v}}\tilde{\kappa}=0$. By writing locally $\partial_{\tilde{v}}\tilde{\kappa}= \d_X \tilde{\Gamma}$, we observe that \eqref{eq:metms} and \eqref{metricKSLapp} are locally isometric. The main conclusion is that, as anticipated before, Siklos space-times are not only conformally Brinkmann (i.e., conformal to space-times with a parallel null vector field), but conformal to pp-wave space-times (i.e., conformal to space-times endowed with a parallel spinor).

\end{document}